\documentclass[12pt, draftclsnofoot, onecolumn]{IEEEtran}
\makeatletter
\def\ps@headings{%
\def\@oddhead{\mbox{}\scriptsize\rightmark \hfil \thepage}%
\def\@evenhead{\scriptsize\thepage \hfil \leftmark\mbox{}}%
\def\@oddfoot{}%
\def\@evenfoot{}}

\makeatother
\ifCLASSINFOpdf
\else
\fi

\usepackage{times,amsmath,color,amssymb,epsfig,cite,enumerate,setspace}
\usepackage{subfigure,multirow,bm,stfloats}
\usepackage[para]{threeparttable}

\usepackage{tabularx}
\usepackage{array}
\usepackage{booktabs}
\usepackage{graphicx}
\usepackage{graphics}
\usepackage{pbox}
\usepackage{dsfont}
\usepackage{psfrag}

\newtheorem{Remark}{\it Remark}[section]

\newtheorem{Proposition}{\it Proposition}[section]
\newtheorem{Lemma}{\it Lemma}[section]

\begin{document}
\begin{spacing}{1.75}
%
\title{Multiuser Energy Diversity in Energy Harvesting Wireless Communications}
\author{
Hang Li$^{\dag}$, Chuan Huang$^{\S}$, and Shuguang Cui$^{\dag}$\\
\fontsize{9pt}{9pt}\selectfont$^{\dag}$Department of Electrical and Computer Engineering, Texas A\&M University, College Station, Texas, 77843 USA\\
$^{\S}$University of Electronic Science and Technology of China, Chengdu, Sichuan, 610051 China\\
}
\maketitle

\begin{abstract}
Energy harvesting communication has raised great research interests due to its wide applications and feasibility of commercialization. In this paper, we investigate the {\it multiuser energy diversity}. Specifically, we reveal the throughput gain coming from the increase of total available energy harvested over time/space and from the combined dynamics of batteries. Considering both centralized and distributed access schemes, the scaling of the average throughput over the number of transmitters is studied, along with the scaling of corresponding available energy in the batteries.
\end{abstract}

\section{Introduction}
For conventional systems with constant power supplies, the multiuser diversity can be exploited when multiple users have independently fading channels. When more users present, it is more likely that the scheduler could find a user with a favorable channel condition. Therefore, the sum or average capacity increases as the number of users getting large. The multiuser diversity gain mainly comes from the effective channel gain \cite{DTse2005}, i.e., from $h_i$ to $\max_{1\le i\le N}h_i$, where $h_i$ denotes the channel power gain. In particular, multiuser diversity with random access or random number of users has been studied in \cite{XQin2006,ABN2012}, and the scaling of the throughput over the number of users was shown to be on the order of $O(\log(N)+\log\log(N))$ \cite{XQin2006}.

Obviously, if all users have identical additive Gaussian channels, there is no multiuser diversity gain, given that all signal channels are the same and the transmission power is constant. However, when powered by energy harvesters, transmitters may have different battery levels because the energy harvesting (EH) rates are random. Then, the variation of battery levels among different users may result in a potential throughput gain over the benchmark, i.e., a point-to-point EH communication system.

In this paper, we revisit the concept of multiuser diversity in EH communications. It is in recent years that harvesting energy from ambient energy sources (e.g., solar, wind, or vibration) has been realized, and enjoys wide applications in the next generation of wireless communication systems, e.g., Internet of Things \cite{SS2011} and heterogeneous networks \cite{JAndrews2014}. Compared against systems with the conventional power supplies that convert fossil fuels into electric energy, EH-based systems are not only more environment friendly, but also more cost-effective by cutting down the service provider utility bills \cite{Bert2007}.

Despite the promising potential, there are two major challenges that hold back the operation of EH wireless systems.
\begin{itemize}
  \item {\it EH Uncertainty}. The power generated by EH is non-deterministic in general due to the dynamic and intermittent characteristics of renewable energy sources, which may not provide a stable power supply for the wireless system. This implies that communications may suffer from unreliability due to the random shortage of energy. Some existing works have studied the impact of such uncertainty brought by EH. For example, the authors in \cite{JAndrews2014} studied a heterogeneous network with multiple base stations (BSs) powered by EH solely. The non-outage probabilities of BSs were derived to analyze the availability region of the network. For a large-scale \emph{ad hoc} network, the author in \cite{KHuang2013} defined a notion called transmission probability to capture the portion of time when the sensor node has enough energy and transmits at a constant power level.
  \item {\it EH Constraints}. These new transmission constraints mean that the available energy at the system up to any time is bounded by its accumulatively harvested energy by then. Many existing works have investigated the throughput optimal or suboptimal transmission strategies under EH constraints. For instance, the optimal throughput has been investigated in point-to-point channel \cite{Ozel2011,CKHo2012}, Gaussian relay channel \cite{HuangZhangCui2013}, and multiuser scenario \cite{LiHang2014B}. In \cite{SUlukus2015}, a comprehensive review was provided on the recent development of EH communications, where throughput-optimal power allocations and scheduling policies were thoroughly discussed under various setups.
\end{itemize}

To address the above two issues, we turn to diversity as motivated by lessons learned in conventional wireless systems, but from a new angle. We study multiuser diversity with respect to the energy availability. To facilitate the analysis, we eliminate the effect of fading channel by considering additive white Gaussian noise (AWGN) channel models only. We investigate the scaling of the available energy across all the users and the scaling of average throughput.

Specifically, assuming that the EH rates are i.i.d. across different users and over time, we explore the multiuser diversity gain over AWGN channels under both the centralized and distributed access schemes:
\begin{itemize}
  \item For the centralized case, we first discover the stationary distribution of the overall battery level, and further analyze the asymptotic behavior of the overall battery level when the number of users goes to infinity. We show that both the greedy scheduling, i.e., choosing the user with the highest available energy at each time, and the rate-suboptimal TDMA access schemes, are all able to explore the multiuser energy diversity, where the average throughput increases on a scale of $\log(\mu N)$, with $\mu$ denoting the mean of energy arrival rate and $N$ denoting the number of users.
  \item For the distributed case, the distribution of energy levels is derived as a function of the channel contention probability, and we show that multiuser energy diversity can be efficiently exploited if the contention probability is on the scale of $O\left(\frac{1}{N}\right)$.
\end{itemize}

The rest of the paper is organized as follows. The system model is given in Section \ref{sysmod}. Then, the multiuser energy diversity is discussed under both centralized and distributed access schemes in Section \ref{multiuser}. Finally, the paper is concluded in Section \ref{fin}.

\section{System Model}\label{sysmod}
In a common multiuser scenario, where multiple transmitter-receiver pairs share one channel for communications, the interference, usually described as packet collisions among users, dominates the unreliability of communication, which significantly impairs the system throughput performance. Thus, we are interested in studying how multiuser diversity affects the system throughput.

To eliminate the multiuser diversity imposed by the channel effect and focus on that form the EH effect, an AWGN channel is adopted for each communication link. Moreover, if two or more transmitters transmit at the same time slot, collisions occur and no data get through\footnote{This is a typical channel model for studying medium access protocols \cite{XQin2006}.}, where length of a time slot is unified (such that the power per slot is of the same magnitude as the corresponding energy per time slot). Suppose that at time slot $t$, only the $n$-th transmitter transmits, the received signal $y^{(n)}_t$ is given by
\begin{equation}
  y^{(n)}_t=\sqrt{P^{(n)}_t}x^{(n)}_t+z^{(n)}_t,\label{RXsignalmulti}
\end{equation}
where $P^{(n)}_t$ is the transmission power, $x^{(n)}_t$ is the transmit signal of unit power, and $z^{(n)}_t$ is the circularly symmetric complex Gaussian (CSCG) noise with zero mean and unit variance. The transmission rate over one time slot could be expressed as $\log\left(1+P^{(n)}_t\right)$ \cite{DTse2005}.

We assume that the EH rates among different transmitters are i.i.d., and each transmitter has a battery with infinite battery capacity\footnote{It is worth pointing out that if the transmitter has no battery but with a constant channel, the analysis is similar to the case with a constant power supply but over i.i.d. fading channel.}. Specifically, let $E^{(n)}_t$ denote the EH rate of the $n$-th transmitter at time slot $t$, which is a Bernoulli random variable such that an energy unit arrives with probability $p$. Furthermore, $\left\{E^{(n)}_t\right\}$ are assumed to be also i.i.d. across time. In addition, the transmitter is able to work in an energy-full-duplex fashion \cite{SUlukus2015}, i.e., it can supply and harvest energy at the same time. Let $B^{(n)}_t$ denote the energy level of the $n$-th user at the beginning of time slot $t$. The power for data transmission at each slot follows a greedy strategy, i.e., the transmitter uses all available energy for data transmission when it accesses the channel.

\section{Multiuser Energy Diversity}\label{multiuser}
In this section, we investigate the multiuser energy diversity under the centralized and distributed access schemes, respectively.

\subsection{Centralized Access}\label{centra}
Assume that the central controller is able to know the energy state information (ESI) of all transmitters at the beginning of each time slot. Here, we consider a greedy scheduling: In each time slot, the controller picks the transmitter with the highest energy level. As such, the transmission power can be written as $M_t=\max_{1\le n\le N}\left\{B_{t}^{(n)}\right\}$, and the instantaneous rate is given by
\begin{equation}
  R^{gr}_t(N)=\log\left(1+M_t\right).
\end{equation}
We use $\mu=p$ to denote the mean of EH rate, and $\sigma^2=p(1-p)$ to denote the variance. Noth that our analysis in this subsection is not limited to the Bernoulli energy arrival model; it works for any arrival model with finite mean and variance.

We aim to analyze the stationary asymptotic behavior of $R^{gr}_t(N)$. The key is to understand how $M_t$ behaves with a large $N$ when $t\rightarrow\infty$. First, we quantify the battery levels when $t\rightarrow\infty$. The following lemma provides a clue to discover the distribution of the battery levels.
\begin{Lemma}\label{stable}
  Energy levels of all transmitters are stable, i.e., $\lim_{t\rightarrow\infty}\mathbb{P}\left\{B_{t}^{(n)}=\infty\right\}=0$ for any $n\in\{1,2,\ldots,N\}$.
\end{Lemma}
The proof is given in Appendix A. Note that this lemma also holds for the case when EH rates are only i.i.d. across time, but not i.i.d. across different transmitters.

Following Lemma \ref{stable}, we have the next proposition.
\begin{Proposition}\label{probenergy}
When $\left\{E_t^{(n)}\right\}$ are i.i.d. across different transmitters and over time, under the greedy scheduling policy $\{M_t\}$, there is
\begin{equation}
  \lim_{t\rightarrow\infty}\mathbb{P}\left\{B_{t}^{(n)}=M_t\right\}=\frac{1}{N},
\end{equation}
for any $n\in\{1,2,\ldots,N\}$.
\end{Proposition}
\begin{proof}
  Since the energy levels of all transmitters are stable as $t\rightarrow\infty$ by Lemma \ref{stable}, it follows that each transmitter could be chosen to transmit with non-zero probability. Also, given that EH rates are i.i.d. across different transmitters and over time, we obtain by symmetry that $\lim_{t\rightarrow\infty}\mathbb{P}\left\{B_{t}^{(n)}=M_t\right\}=\frac{1}{N}$ for any $1\le n\le N$.
\end{proof}
\begin{Remark}\label{remarkdis}
  This also implies that the stationary probability that a transmitter achieves the highest energy level among all transmitters is $1/N$. Then, the waiting time for a transmitter to fulfil a transmission satisfies a geometric distribution with parameter $1/N$.
\end{Remark}

In the following, we only keep the transmitter index $n$ when it is necessary for the presentation; otherwise we remove it since all transmitters are identical to our interests. Based on Proposition \ref{probenergy} and Remark \ref{remarkdis}, we obtain the distribution of energy levels at an arbitrary transmitter, which is given as
\begin{equation}
  B_t\overset{d}{\rightarrow} B=\left\{
     \begin{array}{ll}
       E_1, & \hbox{$\frac{1}{N}$;} \\
       E_1+E_2, & \hbox{$\frac{1}{N}\frac{N-1}{N}$;} \\
       E_1+E_2+E_3, & \hbox{$\frac{1}{N}\left(\frac{N-1}{N}\right)^2$;} \\
       \cdots, & \hbox{$\cdots$,}
     \end{array}
   \right.
\end{equation}
as $t\rightarrow\infty$, where the notation $\overset{d}{\rightarrow}$ denotes the convergence in distribution. In other words, we have
\begin{equation}
  B=\sum_{i=1}^{S}E_i,\label{energyleveldis}
\end{equation}
where $S\sim Geo(\frac{1}{N})$. Then, we obtain
\begin{equation}
  M_t\overset{d}{\rightarrow}M=\max_{1\le n\le N}\left\{B^{(n)}\right\}~\hbox{as}~t\rightarrow\infty,
\end{equation}
in which $B^{(n)}$ is from (\ref{energyleveldis}) for transmitter $n$. Next, we first analyze the asymptotic behavior of $M$ when the number of transmitters gets large; and then consider the scaling of the throughput.

\subsubsection{Scaling of energy level}
It is necessary to discover how the energy level $B$ behaves as $N\rightarrow\infty$; then we move on to $M$. In the next lemma, we present the strong law of large numbers (SLLN) for the random sum $B$.

\begin{Lemma}\label{lemmaslln}
Given $\mu=\mathbb{E}[E]<\infty$, the stationary energy level $B$ satisfies
\begin{align}\label{slln}
  \frac{B-\mu S}{S}\overset{N\rightarrow\infty}{\longrightarrow}0~~\hbox{a.s.,}
\end{align}
where $S\sim Geo(\frac{1}{N})$.
\end{Lemma}
The proof is given in Appendix B. Lemma \ref{lemmaslln} also implies that $\mathbb{E}[B]=N\mu$. Next, we present the central limit theorem for the random sum $B$.
\begin{Proposition}\label{propclt}
Given $\{E_t\}$ are i.i.d. and $\mu=\mathbb{E}[E_t]<\infty$, $B$ satisfies
\begin{align}
  \frac{B-\mu S}{\sigma\sqrt{S}}\overset{d}{\rightarrow}X,
\end{align}
as $N\rightarrow\infty$, where $X\sim \mathcal{N}(0,1)$.
\end{Proposition}
The proof is given in Appendix C.

Based on Proposition \ref{propclt}, we obtain that
\begin{align}
\frac{M-\mu S}{\sigma\sqrt{S}}&=\max_{1\le n\le N}\left\{\frac{B^{(n)}-\mu S}{\sigma\sqrt{S}}\right\}\nonumber\\
&\overset{d}{\rightarrow}\max_{1\le n\le N}\left\{X_n\right\},~\hbox{as}~N\rightarrow\infty,\nonumber
\end{align}
where $X_n\sim \mathcal{N}(0,1)$. Moreover, we can further approximate the distribution of $\max_{1\le n\le N}\left\{X_n\right\}$ according to the next lemma \cite{BovierExtreme,PEExtremal}.
\begin{Lemma}\label{extremenormal}
  If $X_n\sim \mathcal{N}(0,1)$ for $1\le n\le N$, the distribution of $Z_N=\max\{X_1,X_2,\ldots,X_N\}$ satisfies
  \begin{equation}
    \mathbb{P}\left\{a_N(Z_N-b_N)<x\right\}\rightarrow \exp(-e^{-x})
  \end{equation}
  as $N\rightarrow\infty$, where $a_N$ and $b_N$ are normalizing variables, which are given as
  \begin{align}
    a_N &= \sqrt{2\ln N}\nonumber\\
    b_N&=\sqrt{2\ln N}-\frac{\ln\ln N+\ln4\pi}{2\sqrt{2\ln N}}.\nonumber
  \end{align}
\end{Lemma}
Based on Lemma \ref{extremenormal}, we obtain the following proposition.
\begin{Proposition}
The optimal transmit power $M$ satisfies
\begin{equation}
  a_N\left(\frac{M-\mu S}{\sigma\sqrt{S}}-b_N\right)\overset{d}{\rightarrow}Y,~\hbox{as}~N\rightarrow\infty,
\end{equation}
where $a_N$ and $b_N$ are given in Lemma \ref{extremenormal}, and the CDF of $Y$ is $\exp(-e^{-x})$ for $x\in(-\infty,+\infty)$.
\end{Proposition}
\begin{proof}
  It can be directly proved by using Proposition \ref{propclt} and Lemma \ref{extremenormal}.
\end{proof}

\subsubsection{Scaling of expected throughput}
With the results about the transmission power $M$, we are now ready to investigate how the average throughput behaves. By Jensen's inequality, an upper bound of the optimal throughput can be derived as
\begin{align}
  \mathbb{E}\left[R^{gr}(N)\right]=&\mathbb{E}\left[\lim_{t\rightarrow\infty}R_t^{gr}(N)\right]
  =\mathbb{E}\left[\log\left(1+M\right)\right]\nonumber\\
  \le&\log\left(1+\mathbb{E}\left[M\right]\right)=\widehat{R}^{gr}(N).
\end{align}
Note that the upper bound $\widehat{R}^{gr}(N)$ could be very tight when $\mathbb{E}\left[M\right]$ is large, and thus we only focus on the behavior of $\widehat{R}^{gr}(N)$. The next lemma \cite{JustinGuassian} is used to bound the mean of $M$.
\begin{Lemma}\label{upbound}
  If $X_n\sim \mathcal{N}(0,1)$ for $1\le n\le N$, then the mean of $Z_N=\max\{X_1,X_2,\ldots,X_N\}$ satisfies
  \begin{align}
    \mathbb{E}[Z_N]\le\sqrt{2\ln N}+o(1)
  \end{align}
  for large $N$, where $o(1)$ denotes the function such that $\lim_{N\rightarrow\infty}o(1)<\epsilon$ for any $\epsilon>0$.
\end{Lemma}
By Lemma \ref{upbound}, when $N$ is large, we have
\begin{equation}
  \mathbb{E}\left[\frac{M-\mu S}{\sigma\sqrt{S}}\right]=\mathbb{E}\left[\max_{1\le n\le N}\left\{X_n\right\}\right]\le \sqrt{2\ln N}+1.
\end{equation}
Therefore, we obtain an approximated upper bound for $\mathbb{E}[M]$, i.e., for large $N$,
\begin{align}
  \mathbb{E}[M]&\lessapprox \mu N+\sigma \mathbb{E}\left[\sqrt{S}\right]\left(\sqrt{2\ln N}+1\right)\nonumber\\
  &=\mu N+o(N).\label{gaptranspower}
\end{align}
Note that this approximation is more accurate if the variable $S$ is deterministic to be $N$. Furthermore, we could bound $\mathbb{E}[M]$ from below such that
\begin{align}
  \mu N\le\mathbb{E}[M],\nonumber
\end{align}
since $\mu N=\mathbb{E}[B]\le \mathbb{E}[M]$. Finally, it follows that
\begin{equation}
  \widehat{R}(N)=O\left(\log(\mu N)\right),
\end{equation}
where $O\left(\log(\mu N)\right)$ denotes the function such that $\lim_{N\rightarrow\infty}\frac{O\left(\log(\mu N)\right)}{\log(\mu N)}<\infty$.

Another centralized scheme considered here is the fixed TDMA, where each user transmits periodically. In this case, the transmission power is $B=\sum_{i=1}^{N}E_i$ for any user, and we have $\mathbb{E}[B]=N\mu$, which implies that the transmission rate grows on the scale of $\log(\mu N)$. Since the gap $o(N)$ of the transmission power in (\ref{gaptranspower}) grows slowly, it is expected that the throughput achieved by TDMA is almost the same as the greedy scheduling when $N\rightarrow\infty$. One of the advantages of TDMA compared to the greedy algorithm is that TDMA has less complexity since the controller does not need to track the energy level of each user. The performance of TDMA will be also numerically validated in Section \ref{discuss}.

\subsection{Distributed Access}
Suppose that the $n$-th user contents for the channel use with probability $q_n$ at the very beginning of each time slot; then the successful contention probability of the $n$-th user is
\begin{equation}
Q_n=q_n\prod_{j\neq n}(1-q_j).\label{sucessprob}
\end{equation}
Here, we assume that channel contention consumes negligible time and energy as we focus on investigating the order-wise throughput performance. If the $n$-th transmitter successfully occupies the channel, it transmits during the current time slot by using all its available energy (i.e., greedy power utilization). Under this access and power control scheme, the average throughput across the whole system is given by
\begin{align}
  R(N)&=\sum_{n=1}^Nq_n\prod_{j\neq n}(1-q_j)\mathbb{E}\left[\log\left(1+B^{(n)}\right)\right]\nonumber\\
  &\le\sum_{n=1}^NQ_n\log\left(1+\mathbb{E}\left[B^{(n)}\right]\right)\nonumber\\
  &=\widehat{R}_d(N)
\end{align}
Again, it is worth noticing that when $\mathbb{E}\left[B^{(n)}\right]$ is large, $R(N)\approx\widehat{R}_d(N)$. Then, we aim to discover the asymptotic behavior of $\widehat{R}_d(N)$.

Following the Bernoulli energy arrival model, the state transition of energy levels is depicted in Fig. \ref{Markovbattery}.
\begin{figure}
  \centering
  \includegraphics[width=3.3in]{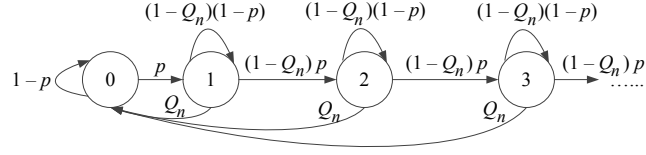}
  \caption{The transition of energy levels.}
  \label{Markovbattery}
\end{figure}
Accordingly, the transition probability matrix of the energy level is given by
\begin{align}
&W=\nonumber\\
&\left[
  \begin{array}{ccccc}
    1-p & p & 0 & \cdots   \\
    Q_n & (1-Q_n)(1-p) & (1-Q_n)p &      \\
    Q_n & 0 & (1-Q_n)(1-p) & (1-Q_n)p   \\
    \vdots &  &  & \ddots   \\
  \end{array}\label{transmatrix}
\right]
\end{align}
We can observe that this Markov chain is irreducible and aperiodic. Moreover, we obtain the stationary distribution of energy levels from the following proposition.
\begin{Proposition}\label{uniquedis}
  There exists a unique stationary distribution $\pi=[\pi_0~\pi_1~\pi_2\cdots]$, where
\begin{align}
  &\pi_0=\frac{Q_n}{p+Q_n}, \label{stationpi0}\\
  &\pi_i=\left(\frac{(1-Q_n)p}{1-(1-Q_n)(1-p)}\right)^i\frac{\pi_0}{1-Q_n},\label{stationpi}
\end{align}
for $i=1,2,\ldots$.
\end{Proposition}
The proof is given in Appendix D.

Next, we analyze the scaling laws of the battery energy and the average throughput. Note that
\begin{equation}
  \lim_{N\rightarrow\infty}Q_n=0.\label{contetionsuc}
\end{equation}
Then, we compute the average energy level as
\begin{align}
  \mathbb{E}[B^{(n)}]=&\frac{p}{1-(1-Q_n)(1-p)}\frac{Q_n}{p+Q_n}\nonumber\\
  &+\sum_{i=2}^{\infty}i\left(\frac{(1-Q_n)p}{1-(1-Q_n)(1-p)}\right)^{i}\frac{\pi_0}{1-Q_n}\nonumber\\
  =&\frac{p}{Q_n+p}\left(\frac{p}{Q_n}+1-p\right)\nonumber\\
  \approx&pQ_n^{-1},
\end{align}
when $Q_n$ is small. If all users apply the same channel contention strategy, it follows that
\begin{align}
  \widehat{R}_d(N)&=\sum_{n=1}^NQ_n\log\left(1+\mathbb{E}\left[B^{(n)}\right]\right)\nonumber\\
  &=NQ_n\log\left(1+\mathbb{E}\left[B^{(n)}\right]\right)\nonumber\\
  &\approx NQ_n\log\left(pQ_n^{-1}\right)\label{rateN}
\end{align}
when $N\rightarrow\infty$. Next, we consider some specific random access strategies and discuss how the multiuser energy diversity can be exploited.

\subsubsection{ALOHA (uniform contention)}\label{unif}
When transmitters contend with probability $q_n=1/N^{\alpha}$, for $\alpha>0$, we obtain $Q_n=\frac{1}{N^{\alpha}}\left(1-\frac{1}{N^{\alpha}}\right)^{N-1}$. The next proposition provides the optimal $\alpha$ which maximizes (\ref{rateN}).
\begin{Proposition}
  Define $\alpha^*$ as
\begin{equation}
  \alpha^*=\underset{\alpha>0}{\arg\max}\widehat{R}_d(N),
\end{equation}
for large $N$. Then, there is $\alpha^*=1$, and the maximum average throughput is given as
\begin{align}
  \widehat{R}_d(N)\approx\frac{1}{e}\log\left(peN\right).\label{purealoha}
\end{align}
\end{Proposition}
\begin{proof}
Note that we have
\begin{align}
  \lim_{N\rightarrow\infty}\left(1-\frac{1}{N^{\alpha}}\right)^{N-1}&=\lim_{N\rightarrow\infty}e^{(N-1)\log\left(1-\frac{1}{N^{\alpha}}\right)}\nonumber\\
  &\approx\lim_{N\rightarrow\infty}e^{-\frac{N-1}{N^{\alpha}}}\nonumber\\
  &=\lim_{N\rightarrow\infty}e^{-N^{1-\alpha}},\nonumber
\end{align}
where the second approximation results from $\lim_{x\rightarrow0}\frac{\log(1+x)}{x}=1$. Thus, we obtain
\begin{align}
  \lim_{N\rightarrow\infty}\left(1-\frac{1}{N^{\alpha}}\right)^{N-1}
  =\left\{
      \begin{array}{ll}
        0, & \hbox{$0<\alpha<1$;} \\
        e^{-1}, & \hbox{$\alpha=1$;} \\
        1, & \hbox{$1<\alpha$.}
      \end{array}
    \right.\nonumber
\end{align}
Next, we check $\widehat{R}_d(N)$ in (\ref{rateN}) for all possible $\alpha$.

When $0<\alpha<1$ and $N$ is large, we obtain
\begin{align}
  \widehat{R}_d(N)&\approx NQ_n\log\left(pQ_n^{-1}\right)\nonumber\\
  &=\frac{N^{1-\alpha}}{e^{N^{1-\alpha}}}\log\left(pe^{N^{1-\alpha}}\right)\rightarrow0\nonumber
\end{align}
as $N\rightarrow\infty$.

When $1<\alpha$, similar to the case $0<\alpha<1$, it can be verified that $\widehat{R}_d(N)\rightarrow0$ as $N\rightarrow\infty$.

When $\alpha=1$, we obtain $Q_n\rightarrow\frac{1}{N}\frac{1}{e}$. It follows that $\mathbb{E}[B]\approx peN$, which leads to (\ref{purealoha}). In all, the proposition is proved.
\end{proof}

\subsubsection{Energy-aware contention}\label{pureenergyaw}
Here, we consider an energy-aware contention such that the transmitter only contends for the channel use when the battery $B\ge pe\log N$, which means that the transmitter acts only when its energy level is higher than a threshold. If the energy level meets the threshold, the transmitter will contend for the channel use with probability $\frac{1}{N}$. Therefore, the overall channel contention probability for user $n$ is given by
\begin{align}
  q_n=&\frac{1}{N}\mathbb{P}\left\{B^{(n)}\ge pe\log N\right\}\nonumber\\
   =&\frac{1}{N}\frac{p}{p+Q_n}\left(\frac{(1-Q_n)p}{1-(1-Q_n)(1-p)}\right)^{pe\log N}.\nonumber
\end{align}

It is expected that the energy-aware contention strategy is strictly better than ALOHA in terms of average throughput. Note that when $N$ is large, it follows that
\begin{align}
  q_n\approx\frac{1-\epsilon}{N},\nonumber
\end{align}
where $\epsilon$ is dependent on $N$. Then, the total number of transmitters that would join channel contentions is $N(1-\epsilon)$. The successful channel contention for user $n$ is given by
\begin{align}
  Q_n=\frac{1-\epsilon}{N}\left(1-\frac{1-\epsilon}{N}\right)^{N(1-\epsilon)-1}\approx\frac{1-\epsilon}{N}e^{-(1-\epsilon)^2},\nonumber
\end{align}
since $\lim_{N\rightarrow\infty}\left(1-\frac{1-\epsilon}{N}\right)^{N(1-\epsilon)-1}=e^{-(1-\epsilon)^2}$. Therefore, we obtain
\begin{align}
  \mathbb{E}[B]\approx \frac{e^{(1-\epsilon})^2}{1-\epsilon}pN,\nonumber
\end{align}
and
\begin{align}
  \widehat{R}_d(N)\approx (1-\epsilon)e^{-(1-\epsilon)^2}\log\left(\frac{e^{(1-\epsilon)^2}}{1-\epsilon}pN\right).\label{imprenergyaware}
\end{align}
Unfortunately, it is extremely difficult to directly prove that when $N$ is large, the average throughput in (\ref{imprenergyaware}) should be strictly larger than that in (\ref{purealoha}). Instead, we numerically verify this result by testing the following two normalized functions based on (\ref{purealoha}) and (\ref{imprenergyaware}):
\begin{align}
  f_1(x)=\log(N),~f_2(x)=\frac{1}{x}\log(xN),\nonumber
\end{align}
for $x\in(0,1]$, where $N$ is a large number such that $N>\frac{1}{x}$ for all chosen $x$. In Fig. \ref{comf1f2}, we draw the values of functions $f_1$ and $f_2$ over $(0,1]$, where the minimum $x$ is set to be 0.001, and $N$ is set to be 1001. Obviously, $f_2$ stays above $f_1$ over the entire region as long as $N$ is large enough.

From the above observation, we conclude that when $N$ is large, the throughput in (\ref{imprenergyaware}) could be strictly larger than that in (\ref{purealoha}). In addition, we also numerically compare the throughput performance of different contention strategies in the next subsection, where it will be shown again that the improved energy-aware contention scheme outperforms the ALOHA uniform scheme.

\begin{figure}
  \centering
  \includegraphics[width=3.7in]{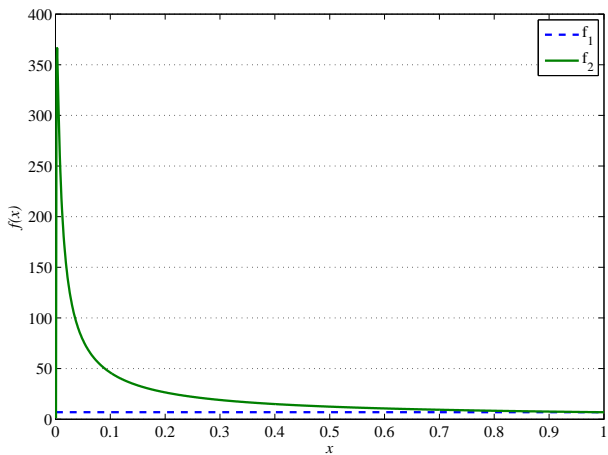}
  \caption{Comparison of functions $f_1$ and $f_2$.}
  \label{comf1f2}
\end{figure}

\subsection{Discussions}\label{discuss}
In this subsection, we provide more insights on the average throughput based on the results in the previous two subsections, and discuss where the multiuser energy diversity gain comes from.

\begin{figure}
  \centering
  \includegraphics[width=3.7in]{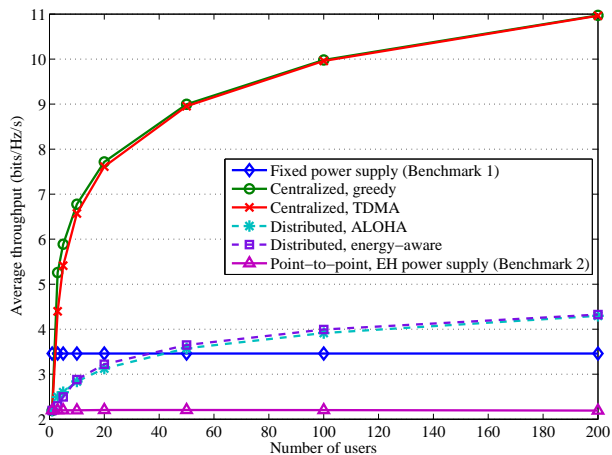}
  \caption{The average throughput in different access schemes.}
  \label{multidivgain}
\end{figure}

First, we numerically compare the average throughput under different centralized and distributed schemes in Fig. \ref{multidivgain}. Here, we set two benchmarks. The first benchmark is the throughput when each user has a fix power supply under the centralized access scheme. This is also equivalent to the point-to-point case since the throughput is always a constant over an AWGN channel given a fixed transmission power. The second benchmark is the throughput achieved by a point-to-point EH communication system over an additive Gaussian channel, where the transmitter adopts a greedy power utilization stratety\footnote{Note that for Gaussian channel, the greedy power utilization strategy is not a capacity achieving power allocation strategy. The capacity achieving power allocation strategy is discussed in \cite{SUlukus2015,RRajesh2014}, and the corresponding throughput is the same as the first benchmark.}. The second benchmark is lower than the first one due to the concavity of the throughput function and the randomness of the transmission power.

We observe that all the scheduling schemes discussed in the previous two subsections can somehow exploit the multiuser energy diversity. For the centralized schemes, the greedy scheduling can achieve better performance than TDMA, while their performances get close when $N$ is large, which agrees with our discussion in Section \ref{centra}. For the distributed schemes, the energy-aware access achieves a slightly higher throughput than ALOHA, which also validates our analysis in the previous section. In addition, we also observe that the distributed scheme has a throughput loss against centralized schemes as $N\rightarrow\infty$, which results from the channel contentions in random access schemes. This observation is similar to the case with conventional multiuser diversity in fading channels, where the ALOHA has a throughput loss $\frac{1}{e}$ compared to the centralized protocol \cite{XQin2006}.

\begin{figure}
  \centering
  \includegraphics[width=3.7in]{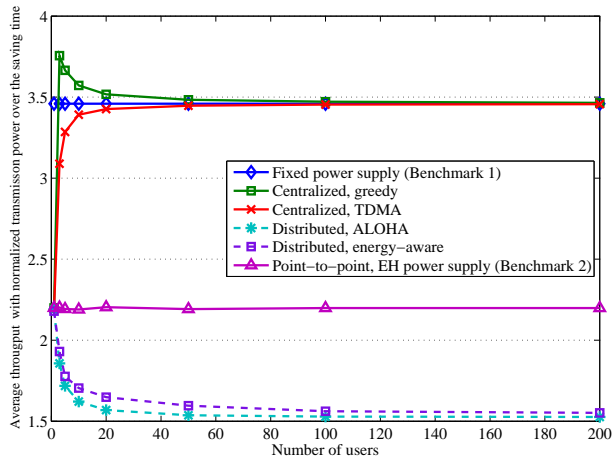}
  \caption{The average throughput with normalized transmission power in different access schemes.}
  \label{multidivgainnorma}
\end{figure}

Moreover, in Fig. \ref{multidivgainnorma}, we numerically compare the average throughput when the transmission power is normalized by the average waiting time, which is $N$. Such normalization eliminates the throughput contribution of the increase of total available energy accumulated over time. We observe that only the centralized greedy scheduling can achieve a throughput gain over both benchmarks, and TDMA only has a gain over the second benchmark while it can approach the first benchmark. This implies that compared to the second benchmark, the multiuser energy diversity gain comes from two aspects:
\begin{enumerate}
  \item The increase in total available energy accumulated over time;
  \item The improvement in effective transmission power:
\begin{itemize}
  \item for the greedy scheduling, improved from $E^{(1)}_t$ to $\max_{1\le n\le N}\frac{1}{N}B^{(n)}_t$;
  \item for TDMA, improved from $E^{(1)}_t$ to $\frac{1}{N}\sum_{t=1}^NE^{(1)}_t$.
\end{itemize}
\end{enumerate}
Note that the normalized average throughput of two centralized schemes is ``upper-bounded'' by that achieved by a fixed power supply as $N\rightarrow\infty$, which implies that multiuser diversity gain mainly comes from the power gain when $N$ is large. We also observe that the distributed schemes have a ``negative'' diversity gain when we eliminate the effect of energy accumulation. It implies that the ALOHA-based access cannot effectively explore the randomness of energy levels since the users are not coordinated well.

Next, we make some remarks on the asymptotic distributions of the transmission power under either centralized or distributed access schemes. The main result is given in the next proposition.
\begin{Proposition}\label{heavytail}
  When EH rates are i.i.d. across transmitters and over time, the transmission power has a heavy-tailed distribution when $N\rightarrow\infty$ under either the centralized optimal access scheme or the distributed access scheme.
\end{Proposition}
The definition of a heavy-tailed distribution is as follows (see Appendix 5 in \cite{AsmussenApp}): The random variable $B$ has a heavy-tailed distribution if
\begin{align}
  \lim_{x\rightarrow\infty}e^{\lambda x}\mathbb{P}\left\{B>x\right\}=\infty\label{heavytailcond}
\end{align}
for all $\lambda>0$. Thus, the key idea of the proof is to verify that the asymptotic distribution of the transmission power satisfies (\ref{heavytailcond}), and the detailed proof is given as follows.
\begin{proof}
  For the centralized case, it is straightforward to show that the distribution $\exp\left(-e^{-x}\right)$ is ``heavy-tailed''. For the distributed case, we have
\begin{align}
  \mathbb{P}\left\{B>x\right\}&=\left(\frac{(1-Q_n)p}{1-(1-Q_n)(1-p)}\right)^{x}\nonumber\\
  &~~\cdot\frac{ 1-(1-Q_n)(1-p)}{Q_n}\frac{\pi_0}{(1-Q_n)}\rightarrow1\nonumber
\end{align}
as $N\rightarrow\infty$, for $x>2$, due to (\ref{contetionsuc}). Therefore, we have
\begin{align}
  &\lim_{x\rightarrow\infty}\lim_{N\rightarrow\infty}e^{\lambda x}\mathbb{P}\left\{B>x\right\}\nonumber\\
  =&\lim_{x\rightarrow\infty}e^{\lambda x}=\infty\nonumber
\end{align}
for any $\lambda>0$, which proves the proposition.
\end{proof}
\begin{Remark}
Proposition \ref{heavytail} considers the probability of ``rare event'' that some transmitter has a very high instantaneous transmission power, which leads to a burst throughput.
\end{Remark}

\section{Conclusions}\label{fin}
In this paper, the multiuser energy diversity gain was investigated. For centralized access schemes, it was shown that the average throughput increases on a scale of $\log(\mu N)$, and the multiuser diversity gain comes from two aspects: the increase of total available energy accumulated over time; and the improvement in effective transmission power. Under the distributed access schemes, the average throughput could increase as well when the access strategy is carefully designed.

\section*{Appendices}

\subsection{Proof of Lemma \ref{stable}}
We prove this proposition by contradiction. Suppose that transmitter $1$ does not satisfy the condition, i.e., $\lim_{t\rightarrow\infty}\mathbb{P}\left\{B_{t}^{(1)}=\infty\right\}>0$. Note that such an event will happen only when transmitter $1$ keeps saving for an infinite number of time slots starting from, say, the $k_1$-th time slot, given the condition that the EH rate has finite nonnegative mean $\mu$ and variance $\sigma^2$. That is, as $t\rightarrow\infty$, we have
  \begin{align}
     \left\{B_{t}^{(1)}=\infty\right\}
     \Leftrightarrow\left\{\sum_{i=k_1}^{t-1}E_{i}^{(1)}=\infty\right\}.\nonumber
  \end{align}
  Moreover, if the event $\left\{\sum_{i=k_1}^{t}E_{i}^{(1)}=\infty\right\}$ happens as $t\rightarrow\infty$, according to the access scheme $M_t$, it is equivalent to the event that the energy level of transmitter $1$ is never the highest among those of all transmitters after time $k$, i.e.,
  \begin{align}
    \left\{\sum_{i=k_1}^{t-1}E_{i}^{(1)}=\infty\right\}
    \Leftrightarrow\left\{\sum_{i=k_1}^{t-1}E_{i}^{(1)}\le \max_{n\neq1}\left\{B_{t}^{(n)}\right\}=\infty\right\}\nonumber
  \end{align}
  as $t\rightarrow\infty$. Then, if the event $\left\{\max_{n\neq1}\left\{B_{t}^{(n)}\right\}=\infty\right\}$ happens, there must exist at least one transmitter, say the $2$-nd transmitter, such that it starts saving from time $k_2$ for an infinite number of time slots, i.e.,
  \begin{align}
    &\left\{\sum_{i=k_1}^{t-1}E_{i}^{(1)}\le \max_{n\neq1}\left\{B_{t}^{(n)}\right\}=\infty\right\}\nonumber\\
    \Rightarrow&\left\{\sum_{i=k_1}^{t}E_{i}^{(1)}\le \sum_{i=k_2}^{t-1}E_{i}^{(2)}=\infty\right\}~\hbox{as}~t\rightarrow\infty.\nonumber
  \end{align}
  Similar to the case of transmitter $1$, if the $2$nd transmitter also saves for an infinite number of time slots, there must be
  \begin{align}
    \left\{\sum_{i=k_2}^{t-1}E_{i}^{(2)}\le \max_{n\neq1,2}\left\{B_{t}^{(n)}\right\}=\infty\right\}~\hbox{as}~t\rightarrow\infty.\nonumber
  \end{align}
  Analogously, it directly implies that all $N$ transmitters must keep saving energy for infinite numbers of times slots. However, this cannot happen since by using the optimal access $\{M_t\}_{t\ge1}$, a transmitter is chosen to fulfil a transmission in each time slot. Hence, all $N$ transmitters cannot keep saving energy forever, which contradicts the assumption that the event $\left\{B_{t}^{(1)}=\infty\right\}$ exists as $t\rightarrow\infty$. Therefore, the lemma is proved.

\subsection{Proof of Lemma \ref{lemmaslln}}
We need to show that for $\forall\epsilon>0$,
\begin{equation}
  \mathbb{P}\left\{\lim_{N\rightarrow\infty}\left|\frac{B-\mu S)}{S}\right|>\epsilon\right\}=0.
\end{equation}
Let $X_i=E_i-\mu$. Note that SLLN holds for $X_1,X_2,\ldots,X_k$, i.e., $\sum_{i=1}^kX_i/k\rightarrow0$ as $k\rightarrow\infty$ with probability 1, which implies
\begin{equation}
  \sum_{k=1}^{\infty}\mathbb{P}\left\{\left|\sum_{i=1}^{k}X_i\right|>k\epsilon\right\}<\infty.
\end{equation}
Define
\begin{align}
   A_{k}=\left\{\left|\sum_{i=1}^{S}X_i\right|>S\epsilon, S=k\right\};~
      F_N=\bigcup_{k\ge N}A_k.\nonumber
\end{align}
Then, we have
\begin{align}
  \mathbb{P}\left\{\lim_{N\rightarrow\infty}\left|\frac{B-\mu S}{S}\right|>\epsilon\right\}
  =\mathbb{P}\left\{\bigcap_{N=1}^{\infty}F_N\right\}
  =\mathbb{P}\left\{A_k~\hbox{i.o.}\right\},\nonumber
\end{align}
where i.o. stands for ``infinitely often''. Next, we need to show $\mathbb{P}\left\{A_k~\hbox{i.o.}\right\}=0$.
\begin{align}
  &\sum_{k=1}^{\infty}\mathbb{P}\left\{A_k\right\}
  =\sum_{k=1}^{\infty}\mathbb{P}\left\{\left|\sum_{i=1}^{k}X_i\right|>k\epsilon\mid S=k\right\}\mathbb{P}\left\{S=k\right\}\nonumber\\
\le&\sum_{k=1}^{\infty}\mathbb{P}\left\{\left|\sum_{i=1}^{k}X_i\right|>k\epsilon\right\}<\infty.\nonumber
\end{align}
Therefore, $\mathbb{P}\left\{A_k~\hbox{i.o.}\right\}=0$, which implies that the convergence (\ref{slln}) holds by the Bore-Cantelli lemma \cite{WilliamsProb}.

\subsection{Proof of Lemma \ref{propclt}}
Let
\begin{equation}
  \frac{B-\mu S}{\sigma\sqrt{S}}=\sum_{i=1}^{S}\frac{E_i-\mu}{\sigma\sqrt{S}}=\sum_{i=1}^{S}\frac{Y_i}{\sqrt{S}}
\end{equation}

Then, we calculate its characteristic function as
\begin{align}
  &\mathbb{E}\left[\exp\left(t\sum_{i=1}^{S-1}\frac{Y_i}{\sqrt{S}}\right)\right]
  =\mathbb{E}\left[\prod_{i=1}^{S}\exp\left(\frac{Y_i}{\sqrt{S}}\right)\right]\nonumber\\
  =&\sum_{s=1}^{\infty}\mathbb{E}\left[\left.\prod_{i=1}^{s}\exp\left(\frac{Y_i}{\sqrt{s}}\right)\right|S=s\right]
     \mathbb{P}\left\{S=s\right\}\nonumber\\
  =&\sum_{s=1}^{\infty}\left(\mathbb{E}\left[\exp\left(\frac{Y_i}{\sqrt{s}}\right)\right]\right)^{s}
     \mathbb{P}\left\{S=s\right\}\nonumber\\
  =&\sum_{s=1}^{\infty}\left(1-\frac{t^2}{2s}+o\left(\frac{t^2}{s}\right)\right)^{s}
     \mathbb{P}\left\{S=s\right\}.\label{charac}
\end{align}
Note that for a large $s$, we have the approximation: $\left(1-\frac{t^2}{2s}+o\left(\frac{t^2}{s}\right)\right)^{s}\approx e^{-\frac{t^2}{2}}$ when $s\ge K$. Thus, we obtain
\begin{align}
  (\ref{charac})=&\sum_{s=1}^{K-1}\left(1-\frac{t^2}{2s}+o\left(\frac{t^2}{s}\right)\right)^{s}
     \mathbb{P}\left\{S=s\right\}\nonumber\\
    &+\sum_{s\ge K}e^{-\frac{t^2}{2}}\mathbb{P}\left\{S=s\right\}.\label{charac2}
\end{align}
Further, by letting $N\rightarrow\infty$, we have
\begin{align}
  \lim_{N\rightarrow\infty}(\ref{charac2})=&\lim_{N\rightarrow\infty}\sum_{s=1}^{K-1}\left(1-\frac{t^2}{2s}+o\left(\frac{t^2}{s}\right)\right)^{s}
     \mathbb{P}\left\{S=s\right\}\nonumber\\
    &+e^{-\frac{t^2}{2}}\lim_{N\rightarrow\infty}\sum_{s\ge K}\mathbb{P}\left\{S=s\right\}\nonumber\\
  =&e^{-\frac{t^2}{2}}\lim_{N\rightarrow\infty}\mathbb{P}\left\{S\ge K\right\}=e^{-\frac{t^2}{2}}.\nonumber
\end{align}
Thus, we obtain that the characteristic function of $\frac{B-\mu S}{\sigma\sqrt{S}}$ converges to $e^{-\frac{t^2}{2}}$ as $N\rightarrow\infty$. Finally, by the L$\acute{e}$vy's continuity theorem (Chapter 18 in \cite{WilliamsProb}), we obtain the conclusion.

\subsection{Proof of Proposition \ref{uniquedis}}
The model given by Fig. \ref{Markovbattery} is an Markov chain with an infinite countable state space, and it has a unique stationary distribution if and only if it has at least one positive recurrent state according to Theorem 26.3 in \cite{Markovinfi}. However, it is difficult to directly show that a state is positive recurrent. Thus, we first derive the form of the stationary distribution $\pi=[\pi_1~\pi_2~\cdots]$, and then show that it is unique.

Assume $\sum_{i=0}^{\infty}\pi_i=1$. Then, by solving $\pi=\pi W$, where $W$ is the transition probability matrix given by (\ref{transmatrix}), we have
\begin{align}
  \pi_0=\frac{Q_n}{p+Q_n},~
  &\pi_1=\frac{p}{1-(1-Q_n)(1-p)}\pi_0\nonumber\\
  &\pi_2=\frac{(1-Q_n)p}{1-(1-Q_n)(1-p)}\pi_1\nonumber\\
  &\pi_3=\frac{(1-Q_n)p}{1-(1-Q_n)(1-p)}\pi_2 \nonumber\\
  &\cdots \nonumber
\end{align}
Thus, we obtain that $\pi$ is given by (\ref{stationpi0}) and (\ref{stationpi}). Next, we check $\sum_{i=0}^{\infty}\pi_i=1$, which can be verified as follows:
\begin{align}
  \sum_{i=0}^{\infty}\pi_i=&\pi_0+\frac{p}{1-(1-Q_n)(1-p)}\pi_0 + \nonumber\\ &\sum_{i=2}^{\infty}\left(\frac{(1-Q_n)p}{1-(1-Q_n)(1-p)}\right)^{i}\frac{\pi_0}{(1-Q_n)} \nonumber\\
  =&\pi_0+\frac{p}{1-(1-Q_n)(1-p)}\pi_0 +\nonumber\\
  &\frac{(1-Q_n)p^{2}}{1-(1-Q_n)(1-p)}\frac{\pi_0}{Q_n} \nonumber\\
  =&\frac{1}{p+Q_n}\left(Q_n+\frac{Q_np}{Q_n+(1-Q_n)p}\right.\nonumber\\
    &~\left.+\frac{(1-Q_n)p^2}{Q_n+(1-Q_n)p}\right)\nonumber\\
  =&\frac{1}{p+Q_n}\frac{(Q_n+p)^2-Q_np(Q_n+p)}{Q_n+(1-Q_n)p}=1.\nonumber
\end{align}
Thus, $\pi$ is a stationary distribution. We observe that state zero is positive recurrent since $\frac{1}{\pi_0}<\infty$, and thus the stationary distribution $\pi$ is unique by Theorem 26.3 in \cite{Markovinfi}.

\end{spacing}
\end{document}